\documentclass[11pt]{article}

\usepackage{mathrsfs}
\usepackage{fullpage}
\usepackage{setspace}
\usepackage{amsthm}
\usepackage{amsmath, amssymb}
\usepackage[ruled]{algorithm}
\usepackage[noend]{algpseudocode}
\usepackage{enumerate}
\usepackage{graphicx}
\usepackage{color}
\usepackage{boxedminipage}
\usepackage{eqparbox,array}

\newtheorem{theorem}{Theorem}
\newtheorem{lemma}[theorem]{Lemma}

\newtheorem{definition}[theorem]{Definition}
\newtheorem{observation}[theorem]{Observation}
\newtheorem{claim}[theorem]{Claim}

\newenvironment{proofclaim}{\begin{trivlist}
    \item[\hskip\labelsep {\it Proof of Claim}.]}{\QED \end{trivlist}}
\newcommand{\QED}{\hfill $\square$}
\newcommand{\hide}[1]{
}

\newcommand{\OPT}{\mathsf{OPT}}

\newcommand{\argmax}{\mathrm{argmax}}
\newcommand{\argmin}{\mathrm{argmin}}
\newcommand{\LP}{\mathrm{LP}}
\newcommand{\negskip}{\vspace{-6pt}}

\newcommand{\ite}[1]{\mbox{(\textbf{#1})}}
\newcommand{\NL}{N^{L+}}
\newcommand{\NN}{N^+}
\newcommand{\calC}{\mathcal{C}}
\newcommand{\calF}{\mathcal{F}}

\newcommand{\zerol}{\{0,L\}}
\renewcommand{\subset}{\subseteq}

\begin{document}
\title{Centrality of Trees for Capacitated $k$-Center}
\author{
Hyung-Chan An \\ EPFL, Switzerland \\  hyung-chan.an@epfl.ch
\and
Aditya Bhaskara \\EPFL, Switzerland \\  aditya.bhaskara@epfl.ch
\and 
Ola Svensson \\ EPFL, Switzerland \\  ola.svensson@epfl.ch
}
\date{}
\maketitle 
\thispagestyle{empty}

\vspace{1ex}

\begin{abstract}
There is a large discrepancy in our understanding of uncapacitated and capacitated versions of
network location problems. This is perhaps best illustrated by the classical $k$-center problem:
there is a simple tight $2$-approximation algorithm for the uncapacitated version whereas the first
constant factor approximation algorithm for the general version with capacities was only recently
obtained by using an intricate rounding algorithm that achieves an approximation guarantee in the
hundreds.

Our paper aims to bridge this discrepancy. For the capacitated $k$-center problem, we give a simple
algorithm with a clean analysis that allows us to prove an approximation guarantee of $9$. It uses
the standard LP relaxation and comes close to settling the integrality gap (after necessary preprocessing),
which is narrowed down to  either $7,8$ or $9$. The algorithm proceeds by first reducing to special {\em tree instances}, and then solves such instances optimally. Our concept of tree instances is quite
versatile, and applies to natural variants of the capacitated $k$-center problem for which we also
obtain improved algorithms. Finally, we give evidence to show that more powerful preprocessing could lead to better algorithms, by giving an approximation
algorithm that beats the integrality gap for instances where all non-zero capacities are uniform.

\vspace{5ex}


\end{abstract}

\newpage
\setcounter{page}{1}

\section{Introduction}\label{sec:intro}
Network location problems lie at the heart of combinatorial optimization. The question of
study is how to select centers so as to best serve a given set of clients located in a metric space. One
can imagine several objective functions to measure the quality of service. Perhaps the most natural
and well-studied ones are ``social welfare'', where we wish to \emph{minimize the average}
distance from a client to its assigned center, and ``fairness'', in which we wish to \emph{minimize
the  maximum} distance from a client to its assigned center. Note that, once we have selected
the centers, both these objectives are minimized by assigning each client to its closest center.
An inherent drawback of this strategy, however, is that it is unable to deal with  centers of
(different) capacities that limit the amount of clients they can serve, which is a constraint present
in most conceivable applications. In fact, these innocent looking capacity contraints have troubled
researchers for decades and they have a much bigger impact on our understanding than the choice of
objective function.

For uncapacitated network location problems, several beautiful algorithmic techniques, such as
LP-rounding~\cite{CharikarGTS02}, primal-dual framework~\cite{JV01} and local
search~\cite{KorupoluPR00,CharikarG05} have been used to obtain a fine-grained understanding of
the approximability of  the
classic variants: $k$-center, $k$-median, and facility location\footnote{Recall that in $k$-center
  and $k$-median, we wish to select $k$ centers so as to minimize the fairness and social
  welfare, respectively; facility location is similar to $k$-median but instead of having a
  constraint $k$ on the number of centers to open, each center has an opening cost.}. Already in the
80's, Gonzales~\cite{Gonzalez85} and Hochbaum \& Shmoys~\cite{HS85} developed tight
$2$-approximation algorithms for the $k$-center problem. For facility location, the current best
approximation algorithm is due to Li~\cite{Li11}. He combined an algorithm by Byrka~\cite{Byrka07}
and an algorithm by Jain, Mahdian, and Saberi~\cite{JainMS02} to achieve an approximation guarantee of $1.488$.
This is nearly tight, as it is hard to approximate the problem within a factor of
$1.463$~\cite{GuhaK99}. The gap is slightly larger for $k$-median: a recent LP rounding~\cite{LS12}
achieves an approximation guarantee of $1+\sqrt{3} \approx 2.732$ improving upon a local search algorithm by Arya et
al.~\cite{AryaGKMMP04}; and it is NP-hard to do better than $1+2/e \approx 1.736$~\cite{JainMS02}. Although the
different problems have algorithms with  different approximation guarantees, they share many techniques, and improvements have often come hand in hand. In particular,  most of the  above progress relies on standard linear programming (LP) relaxations.

In contrast, the standard LP relaxation fails to give any guarantees for capacitated network location problems leading to a much coarser understanding. Apart from special cases, such as
uniform capacities~\cite{KhullerS00}, soft capacities (a center can be opened several
times)~\cite{ShmoysTA97,KhullerS00,JV01}, and other variants~\cite{LeviSS04,ChuzhoyR05}, the only
known constant factor approximation algorithm until recently, was for facility location. In a
sequence of works, including Korupolu, Plaxton \& Rajaraman~\cite{KorupoluPR00}, P{\'a}l, Tardos \& Wexler~\cite{PalTW01}, and Chudak \& Williamson~\cite{ChudakW05}, increasingly enhanced local search algorithms culminated in an
approximation guarantee of 5~\cite{BansalGG12}. Their methods are elegant but specialized to
facility location and are not LP-based. In fact, finding a relaxation-based algorithm for
capacitated facility location with a constant approximation guarantee remains a major open problem
(see e.g. ``Problem 5'' of the ten open problems from the recent book by Williamson and Shmoys~\cite{WS11}). One of the
motivations for finding algorithms based on relaxations is that those methods are often flexible and the
developed techniques transfer to different settings, as has indeed been the case in the study of
uncapacitated location problems.

In the quest to obtain a better understanding and more general (relaxation based) techniques for
capacitated network location problems, it is natural to start with the capacitated $k$-center
problem. Indeed, even though we have a good understanding of uncapacitated location problems in
general, the uncapacitated $k$-center problem stands out, with an extremely simple greedy algorithm that gives a
tight analysis of the LP relaxation. Our failure to understand the capacitated $k$-center
problem is therefore solely due to the lack of techniques for analyzing capacity constraints.
An important recent development in this line of research is due to Cygan, Hajiaghayi and
Khuller~\cite{CyganHK12}, who obtain the first constant factor approximation for the capacitated
$k$-center problem. Their algorithm works by preprocessing the instance to overcome the unbounded
integrality gap of the natural LP relaxation, followed by an intricate rounding procedure. The
approximation factor is not computed explicitly, but is estimated to be roughly in the hundreds.
This however, is still quite far off from the integrality gap of $7$ (after preprocessing)~\cite{CyganHK12} and the
inapproximability results which rule out a factor better than $3$~ (see e.g.~\cite{CyganHK12} for a simple proof).

In this paper, we develop novel techniques to further close the gap in our understanding of
capacitated location problems. In particular, we present a simple algorithm for the capacitated
$k$-center problem with a clean analysis that allows us to prove an approximation guarantee of $9$.
Our result is based on the standard LP relaxation and it almost settles its integrality gap (after
the preprocessing of Cygan et al.~\cite{CyganHK12}): it is either $7,8$ or
$9$ (both the integrality gap and approximation ratio can only take integral values; this is because the worst instances can easily be seen to be ones defined by the shortest-path metric on an unweighted graph). We next describe this and our other
results in greater detail. Due to the simplicity of our analyses, we hope that some of the ideas could
be applied to other location problems, such as capacitated $k$-median, for which no constant factor
approximation algorithms are known.

\paragraph{Our main results and proof outline.}
Our main algorithmic result is the following.
\begin{theorem}
There exists a $9$-approximation algorithm for the capacitated $k$-center problem.
\end{theorem}

Our algorithm takes a guess $\tau$ on the optimal solution value, and considers an unweighted graph $G_{\leq\tau}$ on the given set of vertices where two vertices are adjacent if and only if their distance is at most $\tau$: this graph represents which assignments are ``admissible'' with respect to $\tau$. We solve the standard LP on this graph, which can be assumed to be connected~\cite{CyganHK12}. This determines if it is possible to (fractionally) open $k$ vertices while assigning every vertex to a center that is adjacent in $G_{\leq\tau}$. If this LP is infeasible, we know that the optimum is worse than $\tau$; otherwise, our algorithm will find a solution where every vertex is assigned to a center that is within a distance of 9 in $G_{\leq\tau}$, leading to a 9-approximation algorithm.

The LP solution specifies a set of \emph{opening variables} that indicate the fraction to which each vertex is to be opened. Our algorithm rounds these opening variables by ``transferring'' openings between vertices to make them integral. Since we do not create any new opening, our rounding will naturally open at most $k$ centers; however, the challenge is to ensure that there exists a small-distance assignment of the vertices to open centers. If, for example, the opening of a vertex $v$ is transferred to another vertex that is far away, the clients that were originally assigned to $v$ may be unable to find an available center nearby. For another example, if the opening of a high-capacity vertex gets transferred to a low-capacity one, the low-capacity vertex may fail to provide sufficient capacity to cover the vertices in the neighborhood. Thus, we need to ensure that our rounding algorithm transfers openings only in small vicinity, and that ``locally available capacity'' of the graph does not decrease. (Definition~\ref{def:transfer} formalizes this concept as a \emph{distance-$r$ transfer}.)

We reduce the rounding problem to the special case of \emph{tree instances}, and present an algorithm that rounds such instances optimally. A tree instance is given by a set of opening variables defined on a rooted tree, where every non-leaf node has an opening variable of $1$. Tree instances are generalizations of caterpillars used by Cygan et al.~\cite{CyganHK12}, which can be considered as tree instances whose non-leaf nodes form a path and have certain degree bounds. Suppose we have a tree instance where the capacities are uniform and there are exactly two leaves $u$ and $v$ each of which is opened by $1/2$, whereas every other vertex is opened by $1$. If $u$ and $v$ are distant, this may appear problematic at a glance as we cannot transfer the opening of one to the other. However, there exists a (unique) path $u,w_1,\ldots,w_m,v$ in the tree, and we can transfer the opening of $1/2$ in a ``chain'' along this path: from $u$ to $w_1$, from $w_1$ to $w_2$, $\ldots$, from $w_m$ to $v$. This idea can in fact be carried through to give an algorithm for capacitated $k$-center when all capacities are equal.

Unfortunately, this chain of transfers causes a problem when the capacities are given arbitrarily: suppose in the previous example that $u$ and $v$ have very high capacities compared to the others. Then we will not be able to transfer the opening of $u$ to $w_1$, since the open centers around $u$ may not be able to provide sufficient capacity to cover the vertices that were originally assigned to $u$. However, from another angle, $w_1$ (or any other non-leaf vertex) is ``wasting'' the budget, since it opens a center while contributing relatively small capacity to the graph. This provides us some ``slack'' in the budget that we can utilize: in this particular example, by transferring an opening of $1/2$ \emph{from} $w_1$ \emph{to} $u$, and the other $1/2$ from $w_1$ to $v$ in a chain, we can successfully round the given instance thanks to the decision of closing $w_1$ which had originally had its opening variable equal to one. This strategy of \emph{closing a fully open center} is quite powerful, yet we need to ensure that its capacity can be accomodated by nearby centers if we want to close it. Thus, the viability of such a strategy tends to depend on several factors, including how its capacity compares to vertices in the neighborhood, which of these vertices are to be opened, and so on -- all decisions which could depend on more and more distant vertices.

In contrast, our algorithm departs from previous works by using a simple \emph{local} strategy that does not depend on distant vertices and applies to \emph{every} non-leaf node. The reason our strategy works locally is that the decision of closing fully open centers is determined using solutions to subinstances, which are solved recursively. This key idea significantly eases the analysis and leads to our optimal algorithm for tree instances. The simplicity of our analysis also helps us more carefully analyze the approximation ratio and extend our techniques to related problems. Section~\ref{sec:tree-routing-alg} formally presents our algorithm to round a tree instance; Appendix~\ref{sec:ext} presents the extensions to two related problems: the \emph{capacitated $k$-supplier problem} and the \emph{budgeted opening problem with uniform capacity}.

Section~\ref{sec:gen2tree} presents our reduction to tree instances. We construct a tree instance on a subset of vertices that are chosen as ``candidates'' to be opened. Non-leaf nodes will be carefully chosen, in order to yield a 9-approximation algorithm. Two adjacent vertices in the constructed tree instance will not necessarily be adjacent in the original graph, but will be in close proximity; hence, if the tree instance can be rounded using short transfers of openings, the original instance can also be rounded using only slightly longer transfers.

\paragraph{More results and future directions.}
In Section~\ref{sec:extfd}, we explore future directions towards a better understanding of the problem. Recall that our algorithm proceeds in three steps: firstly, we preprocess the given instance using the results of Cygan et al.~\cite{CyganHK12}; secondly, we reduce the problem to a tree instance; lastly, we solve this tree instance. Given that our tree rounding algorithm is best-possible, it is natural to seek to improve the first two steps.
The preprocessing step
of Cygan et al{\frenchspacing.} allows us to bring down the integrality gap from unbounded to $9$; however, the
integrality gap after the basic preprocessing is known to be at least $7$~\cite{CyganHK12}, which is
larger than the best known inapproximability result that rules out a better factor than $3$. The instance showing the
integrality gap of $7$ (and also that of the inapproximability result) has a special structure that every capacity is either $0$ or $L$ for some constant $L$. In order to understand the potential
of stronger preprocessing methods, we investigate this \emph{$\zerol$-case} and show that
additional preprocessing and a sophisticated rounding gives a $6$-approximation algorithm. The
interesting fact is that we obtain an approximation ratio which \emph{surpasses} the integrality gap lower bound of $7$ after
basic preprocessing. This raises the natural open question: could there be
preprocessing steps which bring the approximation ratio down to $3$? We could also
ask: do lift-and-project methods (applied to a potentially different formulation) automatically {\em
  capture} these preprocessing steps? We believe that understanding these questions would  also shed light
on approximating capacitated versions of other problems such as facility location and $k$-median.

\section{Preliminaries}\label{sec:prelim}
Given an integer $k$ and a metric distance/cost $c:V\times V\to\mathbb{R}_+$ on $V$ with a capacity function $L: V \to \mathbb{Z}_{\geq 0}$, the \emph{capacitated $k$-center problem} is to choose $k$ vertices to \emph{open}, along with an assignment of every vertex to an open center which minimizes the longest distance between a vertex and the center it is assigned to while honoring the capacity constraints: i.e., no open center $v$ is assigned more vertices than its capacity $L(v)$.

For an undirected graph $G=(V,E)$, $d_G(u,v)$ denotes the distance between $u,v\in V$; $N^+_G(u)$ denotes the set of vertices in the neighborhood of $u$, including $u$ itself: $N^+_G(u):=\{v\mid (u,v)\in E\}\cup\{u\}$. For $U\subset V$, $d_G(v,U)$ denotes the distance from $v$ to $U$: $d_G(v,U):=\min_{u\in U} d_G(v,u)$. $N^+_G(U)$ is a shorthand for $\cup_{u\in U}N^+_G(u)$. When the graph of interest $G$ is clear from the context, we will use $d$ and $N^+$ instead of $d_G$ and $N^+_G$, respectively. Let $\OPT$ denote the optimal solution value.

\paragraph{Reduction to an unweighted problem using the standard LP relaxation.} Our algorithm begins with determining a lower bound $\tau^*$ on the optimal solution value: it makes a guess $\tau$ at $\OPT$, and tries to decide if $\tau<\OPT$. We simplify this problem by considering an unweighted graph that represents which assignments are ``admissible''. Let $G_{\leq\tau}=(V,E_{\leq\tau})$ be the unweighted graph on $V$ (with loops on every vertex) where two vertices are adjacent if and only if their distance is at most $\tau$: $E_{\leq\tau} :=\{(u,v)\mid c(u,v)\leq \tau\}$. Note that a feasible solution of value $\tau$ assigns every vertex to a center that is adjacent in $G_{\leq\tau}$, and conversely, if a solution assigns every vertex to a center that is adjacent in $G_{\leq\tau}$, its value is no greater than $\tau$. For an unweighted graph $G=(V,E)$, the standard LP relaxation $\LP_k(G)$ is the following feasibility LP that fractionally verifies whether there exists a solution that assigns every vertex to an open center that is adjacent in $G$:
\begin{equation*}
\boxed{
\begin{array}{rcll}
            \displaystyle \sum_{u\in V}y_u &=&k             ;               &       \vspace{.5ex}\\
            \displaystyle x_{uv} &\leq& y_u                              , &     \forall u,v\in V ;\vspace{.5ex}\\
            \displaystyle \sum_{v: (u,v) \in E}x_{uv} &\leq& L(u)\cdot y_u  ,\ &     \forall u\in V ;\vspace{.5ex}\\
            \displaystyle \sum_{u: (u,v) \in E}x_{uv} &=&1                    , &     \forall v\in V ;\vspace{.5ex}\\
            \displaystyle 0 \  \leq \  x,y  & \leq& 1                             .&    
\end{array}
}
\end{equation*}
$x_{uv}$ is called an \emph{assignment variable}; $y_u$ is called the \emph{opening variable} of $u$.

However, the integrality gap of this LP, defined as the maximum ratio $\frac{\OPT}{\tau}$ where $\LP_k(G_{\leq\tau})$ is feasible, is unbounded; hence this LP cannot in general estimate $\OPT$ very well. We use the approach of Cygan et al.~\cite{CyganHK12} to address this issue: consider the connected components of $G_{\leq\tau}$; if $\tau\geq\OPT$, a vertex can be assigned only to the vertices in the same connected component. For each connected component $G_i$ of $G_{\leq\tau}$, the algorithm decides the minimum value of $k_i$ for which $\LP_{k_i}(G_i)$ is feasible; if $ \sum_i k_i > k$, this certifies that there exists no solution of value $\tau$ or better ($\tau<\OPT$). Now let $\tau^*$ be the smallest $\tau$ for which the algorithm fails to certify that $\tau<\OPT$; since the algorithm will not be able to provide a certificate for $\tau=\OPT$, we have $\tau^*\leq\OPT$. The algorithm then separately solves the subproblems given by the connected components of $G_{\leq\tau^*}$: given a \emph{connected} graph $G$ for which $\LP_k(G)$ is feasible, our algorithm finds a set of $k$ vertices to open, with an assignment of every vertex to an open center that is within the distance of nine. Note that $d_{G_{\leq\tau^*}}(u,v)\leq 9$ implies $c(u,v)\leq 9\tau^*\leq 9\cdot\OPT$ from the triangle inequality.
\begin{lemma}[Cygan et al.~\cite{CyganHK12}]\label{l:preprocessing}
Suppose there exists an algorithm that, given a connected graph $G$, capacity $L$, and $k$ for which $\LP_k(G)$ is feasible, computes a set of $k$ vertices to open and an assignment of every vertex $u$ to an open center $v$ such that $d(u,v)\leq \rho$ and the capacity constraints are satisfied. Then we can obtain a $\rho$-approximation algorithm for the capacitated $k$-center problem.
\end{lemma}

\paragraph{Distance-$r$ transfers.} The above discussion reduces the task of designing
an approximation algorithm for the capacitated $k$-center problem to that of using a solution $(x,y)$
to $\LP_k(G)$ in order to select $k$ centers so that each vertex in the connected graph $G$ is assigned to a
center in a nearby neighborhood. Simple algebraic manipulations show that the LP solution satisfies $|U|=\sum_{u\in U}\sum_{w:(w,u)\in E}x_{wu} \leq \sum_{w\in N^+(U)} L(w)\cdot y_w$;
note that, if the opening variables $y$ are integral, this exactly corresponds to Hall's condition~\cite{H} and hence we can assign every vertex to an adjacent center. However, the LP solution may open each center only by a small fractional amount; in order to obtain an integral solution, it is
therefore natural to try to aggregate fractional openings of nearby vertices. As different
centers have varying capacities, one difficulty of this approach is that the rounding also needs to ensure that the aggregation does not
decrease the available capacity. Consider a center $u$ of capacity $L(u)$ that is open with
fraction $y_u$; we can view it as a center with the \emph{fractional capacity} of $L(u) \cdot y_u$, because in a sense this is the maximum {\em number} (as a
fraction) of vertices this center serves according to the LP.
Our rounding procedure will open $k$ centers, while ensuring that we can \emph{transfer} the fractional capacity
of each $u$ to one or more of the open centers that are close by (and the performance guarantee is determined by how close these centers are).
The following definition formalizes the notion of a distance-$r$ transfer:
\begin{definition}
\label{def:transfer}
Given a graph $G=(V,E)$ with a capacity function $L : V \to \mathbb{Z}_{\geq 0}$ and $y \in \mathbb{R}^{V}_{+}$, a
vector $y'  \in \mathbb{R}^{V}_{+}$ is a distance-$r$ transfer of  $(G,L,y)$ if
\begin{enumerate}[(\ref{def:transfer}a):]
\item \label{def:transfer1} $\sum_{v \in V} y'_v = \sum_{v \in V} y_v$ and
\item \label{def:transfer2}  $\sum_{v:  d(v,U) \leq r}  L(v) y'_v \geq  \sum_{u\in  U} L(u) y_u $ for all $U \subseteq V$.
\end{enumerate}
If $y'$ is the characteristic vector of $S \subseteq V$, we say $S$ is a distance-$r$ transfer of $(G,L,y)$.
\end{definition}
The given conditions say
that a transfer should not change the total number of open centers,
while ensuring that the total fractional capacity in each small neighborhood does not decrease as a result of this transfer.
We also remark that multiple transfers can be composed: if $y'$ is a distance-$r$ transfer of $(G, L, y)$ and $y''$ is a
distance-$r'$ transfer of $(G,L,y')$ then $y''$ is a distance-$(r+r')$ transfer of $(G,L,y)$.

\begin{lemma}\label{lem:framework-main}
For a graph $G=(V,E)$ with a capacity function $L : V \to \mathbb{Z}_{\geq 0}$, let $(x,y)$ be a feasible solution to $\LP_k(G)$. If $S\subset V$ is a distance-$r$ transfer of $(G,L,y)$, then every vertex $v\in V$ can be assigned to a center $s\in S$ such that $d_G(v,s)\leq r+1$, while ensuring no center is assigned more vertices than its capacity. Moreover, $|S|=k$, and this assignment can be found in polynomial time.
\end{lemma}
\begin{proof}
  Consider the natural bipartite matching problem between $V$ and the multiset of open centers that are duplicated to their capacities: i.e, each center $s\in S$ appears in the
  multiset with multiplicity $L(s)$. Every vertex $v$ in $V$ is connected to every copy of each center $s\in S$ such that $d(v,s) \leq r+1$.  Observe that a matching of cardinality $|V|$ naturally defines an assignment that satisfies the desired properties. We shall now show that there exists
  such a matching by verifying Hall's condition, i.e., that for all $U\subset V$, $|U| \leq \sum_{s\in S:d_G(s,U)\leq r+1} L(s)$.

As was observed earlier, we have $|U| \leq \sum_{w:d_G(w,U)\leq 1} L(w)\cdot y_w$; from Condition~(\ref{def:transfer}\ref{def:transfer2}), $|U| \leq \sum_{w:d_G(w,U)\leq 1} L(w)\cdot y_w \leq \sum_{s\in S:d_G(s,U)\leq r+1} L(s)$. This matching can be found in polynomial time, and $|S|=k$ follows from Condition~(\ref{def:transfer}\ref{def:transfer1}).
\end{proof}

\paragraph{Tree instances.} As was discussed earlier, we solve the general problem via reduction to \emph{tree instances}.
\begin{definition}
A tree instance is defined as a tuple $(T,L,y)$, where $T=(V,E)$ is a rooted tree with the capacity function $L:V \to \mathbb{Z}_{\geq 0}$, and \emph{opening variables} $y \in (0,1]^{V}$ satisfy that $\sum_{v\in V}y_v$ is an integer and $y_v=1$ for every non-leaf node $v\in V$.
\end{definition}

\section{Reducing General Instances to Trees}\label{sec:gen2tree}
In this section, we present the reduction from the capacitated $k$-center problem to tree instances.
\begin{lemma}\label{l:redmain1}
Suppose there exists a polynomial-time algorithm that finds an integral distance-$r$ transfer of a tree instance. Then there exists a $(3r+3)$-approximation algorithm for the capacitated $k$-center problem.
\end{lemma}
Lemma~\ref{l:redmain1} directly follows from Lemmas~\ref{l:preprocessing}, \ref{lem:framework-main}, and \ref{l:redmain2}.

\begin{lemma}\label{l:redmain2}
Suppose there exists a polynomial-time algorithm that finds an integral distance-$r$ transfer of a tree instance. Then there exists an algorithm that, given a connected graph $G=(V,E)$, capacity $L:V\to\mathbb{Z}_{\geq 0}$, and $k\in\mathbb{N}$ for which $\LP_k(G)$ has a feasible solution $(x,y)$, finds an integral distance-$(3r+2)$ transfer of $(G,L,y)$.
\end{lemma}

Our reduction, conceptually, constructs a tree instance by defining a tree on a subset of the vertices that have nonzero opening variables in the LP solution. Adjacent vertices in this tree instance may not necessarily be adjacent in $G$, but will be in close proximity; this establishes that a distance-$r$ transfer of the tree instance can be interpreted as a transfer of short distance in $G$ as well. The opening variables of this tree instance would ideally be set equal to the corresponding LP opening variables. However, recall that one of the crucial characteristics of tree instances is that every internal node has the opening variable of one. Yet, individual opening variables of the LP solution may have values less than one in general; we address this issue by using the clustering due to Khuller and Sussman~\cite{KhullerS00}.
\begin{lemma}[Khuller and Sussman~\cite{KhullerS00}]\label{l:clustering}
Given a connected graph $G=(V,E)$, $V$ can be partitioned into $\{C_v\}_{v\in\Gamma}$ for some set of \emph{cluster midpoints} $\Gamma\subset V$, such that\begin{itemize}
\item there exists a tree $U=(\Gamma,F)$ rooted at $r\in\Gamma$ such that for every $(u,v)\in F$, $d_G(u,v)=3$;
\item for all $v\in\Gamma$, $N^+_G(v)\subset C_v$; and
\item for all $u\in C_v$, $d_G(u,v)\leq 2$.
\end{itemize}
\end{lemma}

Observe that, for every cluster $C_v$, the total opening in the neighborhood of $v$ is at least one: $\sum_{u\in N^+_G(v)}y_u\geq \sum_{u\in N^+_G(v)}x_{uv} =1$ from the LP constraints.
We will aggregate these openings to create at least one vertex with the opening variable of one in each cluster; then each cluster will contribute one ``fully open vertex'' to the tree instance, which will become the non-leaf nodes of the tree. Two non-leaf nodes in the tree instance are made adjacent if and only if their clusters are adjacent in $U$. In order to ensure that the aggregation retains the fractional capacity in the graph (in other words, to satisfy Condition~(\ref{def:transfer}\ref{def:transfer2}) of Definition~\ref{def:transfer}), we will transfer the openings in $N^+_G(v)$ to a vertex with the highest capacity in $N^+_G(v)$. Let $m_v:=\argmax_{u\in N^+_G(v)}L(u)$ denote this vertex.

If $m_u$ and $m_v$ are adjacent in this tree instance, how far can they be in $G$? Recall that $m_u$ and $m_v$ are adjacent if and only if $(u,v)\in F$; hence, $d_G(m_u,m_v)\leq d_G(m_u,u)+d_G(u,v)+d_G(v,m_v)\leq 5$. However, here comes a subtlety: if $m_v$ and $m_w$ are also adjacent in the tree, we would expect $d_G(m_u,m_w)\leq d_G(m_u,m_v)+d_G(m_v,m_w)\leq 10$, whereas a tighter bound shows that $d_G(m_u,m_w)$ in fact never exceeds 8: $d_G(m_u,m_w)\leq d_G(m_u,u)+d_G(u,v)+d_G(v,w)+d_G(w,m_w)\leq 1+3+3+1$. Therefore, a simple abstraction that a tree edge corresponds to a length-5 path in $G$ would lead to a slight slack in the analysis. In order to avoid this issue, we will create an \emph{auxiliary vertex} $a_v$ that is ``almost at the same position'' as the cluster midpoint $v$ for each cluster, and aggregate openings to this auxiliary vertex $a_v$ instead of $m_v$ as we did earlier. We will treat $a_v$ as the \emph{delegate} for $m_v$, in the sense that $a_v$ (in lieu of $m_v$) will be part of our tree instance, and if we decide to open $a_v$ from the tree instance, we will open $m_v$ instead.

\begin{figure}[t]
\begin{center}
  \includegraphics[width=14cm]{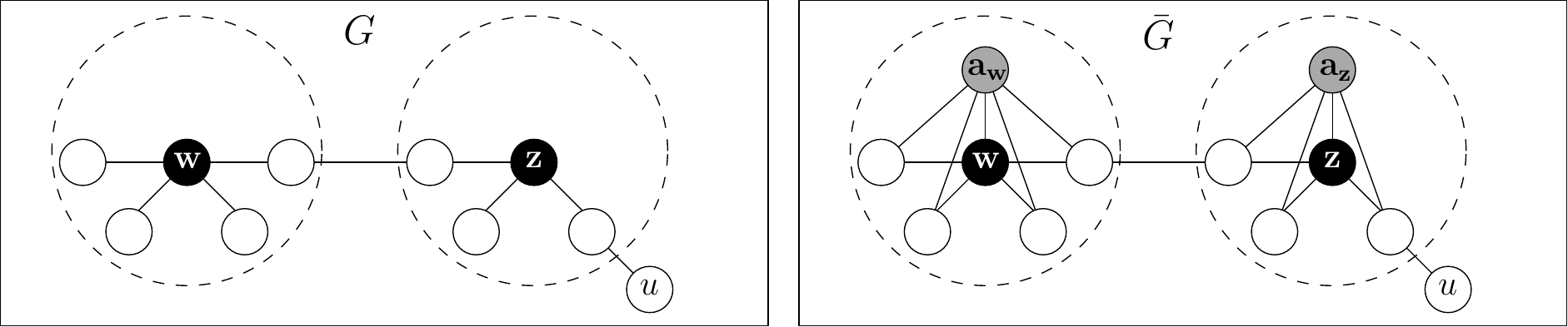} \end{center} 
\caption{Graph $\bar G$ obtained by augmenting $G$ with auxiliary vertices; black nodes correspond to cluster midpoints, dashed circles represent their neighborhoods. }
\label{fig:reduction}
\end{figure}

\begin{proof}[Proof of Lemma~\ref{l:redmain2}]
We first augment the graph by introducing the auxiliary vertices (see also Figure~\ref{fig:reduction}): for each $C_v$, we add a new vertex $a_v$ to the graph, along with the edges from $a_v$ to every vertex in $N^+_G(v)$. Let $\bar G=(\bar V,\bar E)$ be this augmented graph. Observe that $a_v$ is located ``almost at the same position'' as $v$ in the following sense: for every $u\in V$, $d_{\bar G}(u,a_v)=d_G(u,v)$ unless $u=v$; $d_{\bar G}(v,a_v)=1$. Note that $d_{\bar G}(a_w,a_z)=d_G(w,z)$. $L$ and $y$ are accordingly augmented by setting the capacity and the opening variable of the new auxiliary vertex respectively as $L(a_v):=L(m_v)$ and $y_{a_v}:=0$.

Now our reduction works in three phases: in the first phase, we aggregate the opening of 1 from $N_G^+(v)$ to $a_v$; this phase yields a distance-1 transfer $y^{\textsf{first}}$ of $(\bar G,L,y)$. In the second phase, we construct a tree instance by defining a tree on a subset of $\bar V$, and invoke the polynomial-time algorithm to find an integral distance-$r$ transfer of this tree instance. We will see that this transfer can be interpreted as a distance-$3r$ transfer $y^{\textsf{second}}$ of $(\bar G,L,y^{\textsf{first}})$. In the last phase, we transfer the opening of each auxiliary variable $a_v$ to the vertex it delegates, $m_v$. This constitutes a distance-1 transfer $y^{\textsf{third}}$ of $(\bar G,L,y^{\textsf{second}})$.

The opening aggregation in the first phase works as follows: for each cluster $C_v$, we increase $y_{a_v}$ while simultaneously decreasing $y_u$ for some $u\in N^+_G(v)$ with $ y_u>0$. If $ y_{a_v}$ reaches one, we stop; if $ y_u$ reaches zero, we find another $u\in N^+_G(v)$. The initial choice of $u$ is always taken as $m_v$ so that this procedure ensures that $ y_{m_v}$ becomes zero. The procedure outputs a distance-1 transfer $y^{\textsf{first}}$, since whenever an opening variable decreases during the construction, we increase the opening variable of an adjacent vertex with higher or equal capacity.

In the second phase, we define a tree $T$ on the set of vertices with nonzero opening variables. Note that this in particular implies that $m_v\notin T$ for each cluster $C_v$. $T$ is constructed from $U =
(\Gamma,F)$ as follows: we replace each $v\in \Gamma$ by $a_v$ to obtain a tree on the auxiliary vertices, and for each vertex $u\in C_v$ such that $ y_u>0$, we attach $u$ as a (leaf) child of $a_v$. Note that every non-leaf node is an auxiliary vertex and therefore has the opening variable of one. The total opening is equal to the total opening of $y$, and therefore $(T,L,y^{\textsf{first}})$ is a valid tree instance; we invoke the polynomial-time algorithm to find an integral distance-$r$ transfer of this instance. For any two nodes $i$ and $j$ that are adjacent in this tree instance, either $i=a_u$ and $j=a_v$ for some $(u,v)\in F$, or $i=a_v$ and $j\in C_v$. In the former case, $d_{\bar G}(i,j)=3$; in the latter case, $d_{\bar G}(i,j)\leq 2$. Thus, the integral distance-$r$ transfer of the tree instance can be interpreted as an integral distance-$3r$ transfer $y^{\textsf{second}}$ of $(\bar G,L,y^{\textsf{first}})$.

Note that $y_{m_v}^{\textsf{second}}=0$ for every cluster $C_v$, since $m_v$ does not participate in the tree instance; on the other hand, $a_v$ may have been opened by the tree algorithm. In the last phase, we transfer the opening of $a_v$ to $m_v$, the vertex delegated by $a_v$. This yields an integral distance-$1$ transfer $y^{\textsf{third}}$ of $(\bar G,L,y^{\textsf{second}})$.

Note that $y_{a_v}^{\textsf{third}}=0$ for every cluster $C_v$; by projecting $y^{\textsf{third}}$ back to $V$, we obtain an integral distance-$(3r+2)$ transfer of $(G,L,y)$.
\end{proof}

\section{Algorithm for Tree Instances}\label{sec:tree-routing-alg}
In this section we prove the following.
\begin{lemma}\label{l:treealgorithm}
  There is a polynomial time algorithm that finds an integral distance-$2$ transfer of a given tree instance $(T,L,y)$.
\end{lemma}
We remark that it is easy to see that some tree-instances do not admit an integral distance-$1$ transfer and
the above lemma is therefore the best possible. One example is the following: the instance consists
of a root with six children, where each child is opened with a fraction $2/3$, and all vertices have
the same capacity; it is easy to see that any integral solution needs to transfer fractional
capacity from one leaf to another (i.e., of distance $2$).  We now present the algorithm along with
the arguments of its correctness.

The algorithm builds up the solution by recursively solving smaller tree instances.  The base case
is simple: if $|T| \leq 1$ then simply open the vertex in $V(T)$ if any. By the integrality of
$\sum_{v\in V(T)} y_v$ this is clearly a distance-$2$ transfer (actually a distance-$0$
transfer). Let us now consider the more interesting case when $|T| \geq 2$; then there exists a node
$r$ of which every child is a leaf. Let $v_1, \dots, v_\ell$ be the children of $r$, in the
non-increasing order of capacity: $L(v_1)\geq\cdots\geq L(v_\ell)$. Let $T_r$ denote the subtree
rooted at $r$ and $Y:=\sum_{i=1}^{\ell}y_{v_i}$. The algorithm considers two separate cases
depending on whether $Y$ is an integer.

Let us start with the simpler case when $Y$ is an integer: the algorithm selects the set $S_r$
consisting of the $Y+1$ vertices of highest capacity in $T_r$. As every pair of nodes in $T_r$ are
within a distance of $2$, $S_r$ is a distance-2 transfer of the tree instance induced by $T_r$. The algorithm
then solves the tree instance induced by $\bar T: = T \setminus T_r$ to obtain a distance-2 transfer
$\bar S$ of size $\sum_{v\in T} y_v - Y - 1$. It follows that $S:= S_r
\cup \bar S$ is a distance-2 transfer of $(T, L, y)$.

We now consider the final more interesting case when $Y$ is not an integer. In this case, we \emph{cannot} consider
$T_r$ and $T\setminus T_r$ as two separate instances because the $y$-values suggest to
either open $\lfloor Y \rfloor +1$ or $\lceil Y \rceil +1$ centers in $T_r$: a choice that depends
on the selected centers in $T\setminus T_r$.  As at least $\lfloor Y \rfloor +1$ of the vertices
in $T_r$ will be selected as centers in either case, the algorithm will naturally commit itself to
open the $\lfloor Y \rfloor + 1$ vertices in $T_r$ of highest capacity. Let $S_{\mathsf{commit}}$
denote that set and note that it equals $\{v_1, \dots, v_{\lfloor Y \rfloor}, r\}$ or $\{v_1, \dots,
v_{\lfloor Y \rfloor}, v_{\lfloor Y \rfloor +1}\}$ dependent on which node of $r$ and $v_{\lfloor
  Y\rfloor +1}$ has highest capacity ($v_{\lfloor
  Y\rfloor +1}$ is well defined since we have that the number of children $\ell$ is at least $\lceil
Y \rceil$ from $y \leq \mathbf{1}$). 
By the selection of $S_{\mathsf{commit}}$, we have 
\begin{align}
\label{eq:commit}
\sum_{u\in V(T_r)} y_u L(u) \leq \sum_{s\in S_{\mathsf{commit}}} L(s) +   \bar y_p \bar L(p),
\end{align}
where $\bar y_p = Y - \lfloor Y \rfloor$ and $\bar L(p) = \min[L(r), L(v_{\lfloor Y \rfloor + 1})]$.  In other
words, if the algorithm on the one hand chooses to only open the $\lfloor Y \rfloor + 1$ centers $S_{\mathsf{commit}}$ in
$T_r$, then an additional fractional capacity $\bar y_p \bar L(p)$ needs to be transferred from $T_r$ to an open center
in $T\setminus T_r$. On the other hand, if the algorithm chooses to open all the  centers  $\lceil Y \rceil + 1$ in
$S_{\mathsf{commit}} \cup \{v_{\lfloor Y \rfloor + 1}, r\}$ then those centers can
accomodate all the fractional capacity in $T_r$ together with $(1-\bar y_p)\bar L(p)$ additional capacity.
\begin{figure}[t]
\begin{center}
\includegraphics[width=12cm]{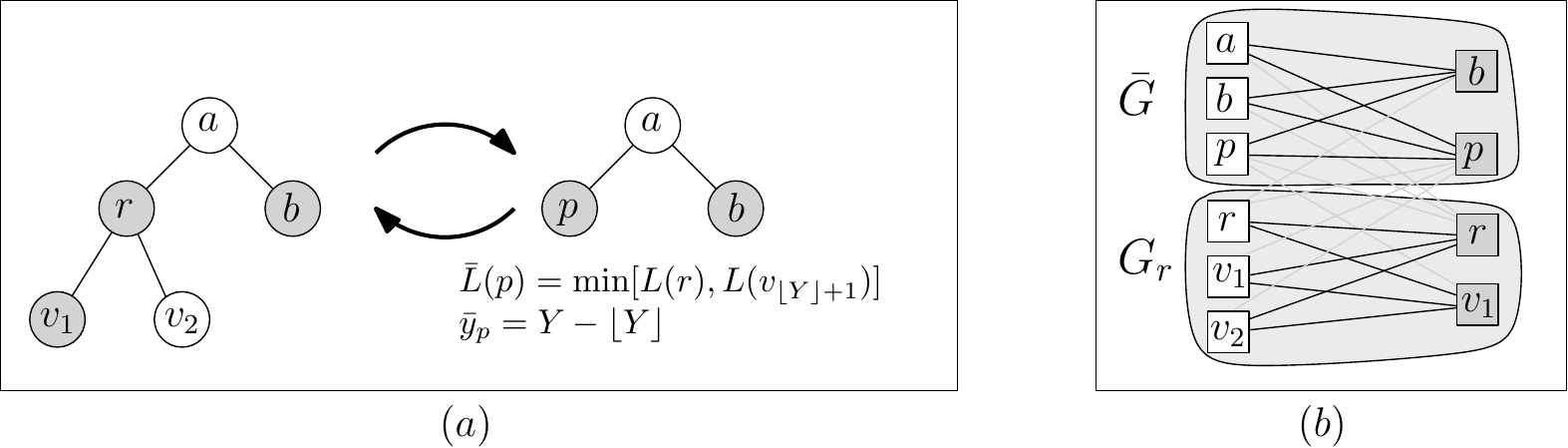}
\end{center}
\caption{(a) The construction of $\bar T$ from $T$ with the subtree $T_r$ rooted at $r$ with
  children $v_1$ and $v_2$; the grey vertices are those selected in potential solutions to $\bar T$ and
  $T$, respectively. (b) The bipartite graph and the induced subgraphs $\bar G$ and $G_r$ that are
  used in the proof of Claim~\ref{claim:treealgo}.}
\label{fig:treealgo}
\end{figure}

We defer this decision to be based on the solution of the smaller tree
instance $(\bar T,  \bar L, \bar y)$ obtained from $(T,  L, y)$ as follows (see also Figure~\ref{fig:treealgo}a): replace  $T_r$ by
the vertex $p$ that represents the deferred decision and
 let  $\bar y, \bar L$ be the natural restrictions of $y, L$ on $ T \setminus T_r$ with  $\bar
y_p = Y - \lfloor Y \rfloor$ and $\bar L(p) =\min[L(r), L(v_{\lfloor Y \rfloor + 1})]$.
The algorithm then recursively solves this smaller instance to obtain a distance-2 transfer $\bar S$
of $\bar T$. From $\bar S$ it constructs the solution $S$ to the original problem instance by first
replacing $p$ by the vertex $v_{\lfloor Y \rfloor +1}$ or $r$ that was not chosen to be  in $S_{\mathsf{commit}}$ if $p\in
\bar S$, and then adding $S_{\mathsf{commit}}$ to it.

We complete the proof of Lemma~\ref{l:treealgorithm} by arguing that $S$ is a distance-2 transfer of the original
tree instance $(T,L,y)$.  Note that, as $|\bar S| = \sum_{v \in \bar T} \bar y_v = \sum_{v\in T} y_v
-1 - \lfloor Y \rfloor$, we have $|S| = |\bar S| + |S_{\mathsf{commit}}| = \sum_{v\in V} y_v$ as required. It remains to verify
Condition~(\ref{def:transfer}\ref{def:transfer2}) of Definition~\ref{def:transfer}:

\begin{claim} 
\label{claim:treealgo}We have
$
\displaystyle \sum_{u \in U} y_u L(u) \leq \sum_{s \in S: d(s, U) \leq 2} 
L(s)$ for all   $U \subseteq V(T).
$
\end{claim}
\begin{proofclaim}
Consider the bipartite graph $G$ with left-hand-side $V(T)$, right-hand-side $S$, and an edge between
$v\in V(T)$ and $s \in S$ if  $d(s,v) \leq 2$. For simplicity, we slightly abuse notation and think of
$V(T)$ and $S$ as disjoint sets. Moreover, let $N(U)$ denote the neighbors
of a subset $U$ of vertices in this graph and let $w: V(T) \cup S \rightarrow \mathbb{R}$  be weights  on the
vertices defined by
$$
w(v) = \begin{cases} y_v L(v) & \mbox{if }  v\in V(T) \\
  L(v) & \mbox{if } v\in S 
\end{cases}. $$
With this notation, we can reformulate  the condition of the claim as
\begin{equation}
\label{eq:graphcond}
\sum_{u \in U} w(u) \leq \sum_{s \in N(U)} 
w(s) \qquad \mbox{ for all } U \subseteq V(T).
\end{equation}
To prove this, we shall prove a slightly stronger statement by verifying the condition separately on
two biparite graphs $G_r$ and $\bar G$ that correspond to $T_r$ and $\bar T$, respectively.  We
obtain $G_r$ and $\bar G$ from $G$ as follows (see also Figure~\ref{fig:treealgo}b). First, add a vertex $p$ to the left-hand-side by
making a copy of $r\in T$ and set $w(p) = \bar y_{p}\cdot \bar L(p)$ and update $w(r) = y_r L(r) -
\bar y_p \bar L(p)  = L(r) - \bar y_p \bar L(p) \geq 0$.
Similarly,  if $p\in \bar S$ then add a copy  $p$ of $r\in S$ and set  $w(p)   = \bar L(p)$ and update
$w(r) = L(r) - \bar L(p) \geq 0 $. Note that after these operations the vertices of both the left-hand-side and
the right-hand-side can naturally be partitioned into those that correspond to vertices in $T_r$ and
those that correspond to vertices in $\bar T$. Graphs $G_r$ and $\bar G$ are the subgraphs induced by these  two
partitions. 

Let us first verify that~\eqref{eq:graphcond} holds for $\bar G$.  By construction, we have that the
total weight $w(U)$ of a subset $U$ of $V(\bar T)$ is equal to $\sum_{u\in U} \bar y_u \bar L(u)$ and the
total weight $w(N(U))$ of its neighborhood in $\bar G$ equals $\sum_{s\in \bar S: d(s, U) \leq 2}
\bar L(s)$. Hence,~\eqref{eq:graphcond} holds since $\bar S$ is a distance-2 transfer of $\bar T$.

We conclude the proof of the claim by verifying~\eqref{eq:graphcond} for $G_r$. As both the
left-hand-side and right-hand-side of $G_r$ correspond to vertices in $T_r$ that all are within
distance $2$ of each other, we have that $G_r$ is a complete bipartite graph.
The total weight of the left-hand-side is by construction $\sum_{u\in T_r} y_u L(u) - \bar y_p \bar
L(p)$ and the total weight of the right-hand-side is $\sum_{s\in T_r \cap S}   L(s) -\bar L(p)
\mathbf{1}_{p \in \bar S}$ which equals $\sum_{s\in S_{\mathsf{commit}}} L(s)$. The claim now
follows from~\eqref{eq:commit}, i.e., that $\sum_{u\in T_r} y_u L(u) - \bar y_p \bar
L(p) \leq \sum_{s\in S_{\mathsf{commit}}} L(s)$.
\end{proofclaim}

The above claim completed the analysis of the algorithm for finding an integral distance-2 transfer of a given
tree instance and Lemma~\ref{l:treealgorithm} follows.

\section{Better preprocessing for better algorithms}\label{sec:better-preprocessing}
In this section, we explore the possibility of a further improvement in the performance guarantee and integrality gap bounds via a better preprocessing. We demonstrate this by presenting a 6-approximation algorithm for the $\zerol$-case of our problem. Formally, this is the special case of the capacitated $k$-center problem in which all the vertex capacities are either $0$ or $L$, for some integer $L$. Instances with this property will be called $\zerol$-instances.

It turns out that instances arising from the NP-hardness results, as well as the gap instances for the standard LP relaxation are all of this form, so this special case seems to capture the essential combinatorial difficulty of the capacitated problem. For these, we prove the following theorem:

\begin{theorem}\label{thm:zerol}
There is a polynomial-time algorithm achieving a $6$-approximation for $\zerol$-instances of the capacitated $k$-center problem.
\end{theorem}

\negskip \negskip
\paragraph{General framework revisited.}
Let us recall the preprocessing done by Cygan et al.~\cite{CyganHK12}, explained in Section~\ref{sec:prelim}. The idea is to guess the optimum (call the guess $\tau$), and consider an unweighted graph $G_{\le \tau}$ in which we place an edge between $u,v$ if $d(u,v) \le \tau$. If we then solve the LP on such a graph, the integrality gap is unbounded, as the following example shows. We have two groups of $3$ vertices, such that the distance within the groups is $1$, and distance between the groups is some large $C$. Suppose the capacity of each vertex is $2$, and $k=3$. Then the LP for the instance is feasible with $\tau=1$, while the $\OPT$ is $C$. 

The trick to avoid this situation is to restrict to connected components of $G_{\le \tau}$ defined above, then for each $i$, determine the smallest $k_i$ for which the LP is feasible in component $i$, and finally check if $\sum_i k_i \le k$. (If not, the guess $\tau$ is too small). For connected graphs $G_{\le \tau}$, our main theorem shows that the integrality gap is at most 9. Cygan et al.~\cite{CyganHK12} gave a connected $\zerol$-instance with integrality gap $7$, i.e., the $\LP_k(G_{\leq\tau})$ is feasible with $\tau=1$, while $\OPT$ $\ge 7$. 

So in a nutshell, the steps above can be seen as coming up with a graph (in this case $G_\le \tau$) which has edges between $u,v$ only if it is feasible to assign $u$ to $v$ and vice versa, solving the LP on the connected components of this graph, and verifying that $\sum_i k_i \le k$, as described. The aim of better preprocessing would then be to come up with a graph with even fewer edges, while still guaranteeing that the optimum assignment is preserved. Intuitively, this could produce more connected components, thus the $\sum_i k_i \le k$ condition now becomes {\em stronger}. 

For the $\zerol$-case, we prove that an extremely simple additional preprocessing -- namely removing edges between vertices $u$ and $v$ with $L(u)=L(v)=0$ -- provably lowers the integrality gap. 
Our result is then the following.

\begin{theorem}\label{t:zerolmain}
Suppose $G_{\le \tau}^*$ is a connected component after the two preprocessing steps above, and suppose $\text{LP}_k(G_{\le \tau}^*)$ is feasible, for some $k$. Then there is an algorithm to compute a set of $k$ vertices to open and an assignment of every vertex $u$ to an open center $v$ such that $d(u,v) \le 6$, and the capacity constraints are satisfied.
\end{theorem}

The preprocessing leads to additional structure in the instance which we then use carefully in our rounding procedure. The proof is presented in Appendix~\ref{sec:ap:zerol}. A natural open question is whether such an approach can be applied to the general problem as well, improving our 9-approximation algorithm.

\section{Extensions to other problems and future directions}\label{sec:extfd}
Our techniques can be extended to obtain approximation algorithms for other problems.
Appendix~\ref{sec:ext} discusses two problems to which our techniques readily apply: first we study the capacitated
$k$-supplier problem -- a variant of $k$-center where the set of clients and facilities are specified separately -- 
and give an 11-approximation algorithm. We then consider the budget generalization of the $k$-center
problem, where the general capacity problem is inapproximable but we give a 9-approximation
algorithm when the capacities are uniform. We see this as further evidence  that the simplicity of our
approach helps in designing better algorithms also for other location problems.

As our $9$-approximation algorithm comes close to settling the integrality gap, it is natural to ask
if our techniques can be used  to obtain a tight result.
 Recall that our framework consists of first reducing the general problem to tree instances and then
solving such instances. Since our algorithm for tree instances is optimal, any potential improvement
must come from the reduction, and we raise this as an open problem. 

Finally, our preliminary results on additional preprocessing indicate that further investigation is
necessary to understand if these techniques can help bring down the integrality gap to the tight
factor of $3$. More generally, we believe that it is important not only for capacitated $k$-center
but also for other problems, such as facility location and $k$-median, to understand the power of
lift-and-project methods (applied to potentially different formulations). For example, do they
automatically capture these preprocessing steps and lead to stronger formulations?

\bibliographystyle{plain}
\bibliography{refs}

\appendix
\section{Extensions to other problems}\label{sec:ext}
We believe that the simplicity of our approach could be key to generalizing it to other location
problems with capacity constraints. In this section, we see how our ideas readily
apply to two problems.

\subsection{Capacitated $k$-supplier}
In this subsection, we present a $11$-approximation algorithm for the capacitated $k$-supplier problem.
This problem is a generalization of the capacitated $k$-center problem in which some vertices
  are designated clients and some facilities. We can only open $k$ of the facilities, and the aim is
  to serve the clients (facilities do not have to be {\em served}).
  
Let us denote by $\calC$ and $\calF$ the set of clients and facilities respectively. For this version, we prove the following.

\begin{theorem}\label{thm:clientfacility}
There exists a polynomial time $11$-approximation algorithm for the capacitated $k$-supplier problem.
\end{theorem}
The algorithm proceeds along the lines of our main result. We first guess the optimum $\tau$, and restrict to the bipartite graph $G$ on vertex sets $\calC, \calF$, with an edge between $u \in \calC$ and $v \in \calF$ iff $d(u,v) \le \tau$. We then divide this into connected components and work with them separately, as before. Thus in what follows, let us assume that $G$ as defined above is connected, and $LP_k(G)$ is feasible. Note that this is a slightly different LP, where $y$-variables exist only for facilities, and the constraints $\sum_{u:(u,v) \in E} x_{uv} =1$ exist only for the clients.

The main difference in this variant is in the clustering step. This now works as follows. Start with a client $u \in \calC$, and include all of $N^+(u)$ in the cluster $C_u$. Now as long as possible, do the following: pick a client $u \in \calC$ which is at a distance $>2$ from the midpoints of all the clusters so far, but is distance precisely $4$ from some cluster midpoint; include all of $N^+(u)$ into the cluster $C_u$ (there will not be an overlap with other clusters because of the distance condition).

When the procedure ends, we will be left with a bunch of clients at distance $2$ from some cluster midpoints, and some facilities at distance $3$ from some cluster midpoints (and nothing else, by connectivity properties). We move them to the closest cluster (breaking ties arbitrarily). Now the procedure satisfies the following conditions:
\begin{enumerate}
\item Each cluster has its $y$-values adding up to $\ge 1$ (indeed, the neighborhood of the cluster midpoint has total $y$-value $\ge 1$, as is required in the tree reduction).
\item The graph of clusters, in which we place an edge if the midpoints are at distance precisely $4$, is connected.
\end{enumerate}

These properties ensure that we can perform precisely the same reduction to tree instances, however we have a variant of Lemma~\ref{l:redmain1}: an $r$-transfer to the tree instance now implies a $(4r+3)$ approximation algorithm for the client/facility problem. This is because adjacent cluster midpoints are at a distance $4$, and hence the distance in $G$ between two vertices $a_u$ and $a_v$ (as in the reduction) which have distance $r$ in the tree instance, is now $4r$. The rest of the proof carries over verbatim, and we obtain a reduction to tree instances with the above guarantee.

This proves Theorem~\ref{thm:clientfacility}, because for tree instances, we can use our algorithm which gives $r=2$. \qed

\subsection{Budgeted version with uniform capacities}

The \emph{budgeted center problem} is a weighted generalization of the $k$-center problem: in the $k$-center problem, opening a center incurs the uniform cost of one and there is a budget of $k$ on the total opening cost; on the other hand, in the budgeted center problem, the opening costs are given by $C:V\to\mathbb{R}_+$ that is a part of the input along with the total budget $B\in\mathbb{R}_+$. It is NP-hard to approximate this problem to any approximation ratio if the vertices have general capacity; this can be shown by a straightforward reduction from the Knapsack Problem. However, for the uniform capacity, Khuller and Sussmann~\cite{KhullerS00}, using the technique of Bar-Ilan, Kortsarz, and Peleg~\cite{BKP}, gives a 13-approximation algorithm. In this subsection, we present a 9-approximation algorithm for the budgeted center problem with uniform capacities. We note that it is easy to extend this result to the $\zerol$-case as well.

Let $L_0\in\mathbb{N}$ be the uniform capacity. Following is the key lemma of our analysis:
\begin{lemma}\label{l:budgetmain}
Suppose there exists a polynomial-time algorithm that finds an integral distance-$r$ transfer of a tree instance. Then there exists an algorithm that, given a connected graph $G=(V,E)$, the constant capacity function $L:V\to\{L_0\}$, and $k\in\mathbb{N}$ for which $\LP_k(G)$ has a feasible solution $(x,y)$, in addition to the opening costs $C:V\to\mathbb{R}_+$, finds an integral distance-$(3r+2)$ transfer $y'$ of $(G,L,y)$ satisfying $\sum_{v\in V}C(v)y'_v \leq \sum_{v\in V}C(v)y_v$.
\end{lemma}

Our 9-approximation algorithm follows from Lemma~\ref{l:budgetmain}.
\begin{theorem}\label{t:budgetmain}
There exists a 9-approximation algorithm for the budgeted center problem with uniform capacities.
\end{theorem}
\begin{proof}
Let $\OPT$ denote the optimal solution value. As in Lemma~\ref{l:preprocessing}, our algorithm makes a guess $\tau$ at the optimal solution value and tries to decide if $\tau<\OPT$. In this problem again, we consider the graph $G_{\leq\tau}$ representing the admissible assignments. Consider the connected components of $G_{\leq\tau}$; for each component $G_i$, we will compute a lower bound $B_i$ on the minimum budget necessary to have a feasible solution to the subproblem induced by $G_i$. Observe that, if $\tau\geq\OPT$, an optimal solution assigns every vertex to a center that is in the same connected component. Thus, if $\sum_i B_i > B$, we can certify that $\tau<\OPT$. $B_i$ is determined by solving $\LP_{k_i}(G_i)$, but with an objective of minimizing the opening cost $\sum_{v\in G_i} C(v)y_v$ rather than as a feasibility LP with no objective function; $k_i$ is chosen by trying all integers from 1 to $|V(G_i)|$ and selecting the one that gives the smallest opening cost. If we failed to certify $\tau<\OPT$, this means $\sum_i B_i \leq B$. Now for each $G_i$, Lemmas~\ref{lem:framework-main} and \ref{l:budgetmain} lets us find a set of vertices to open for which there exists an assignment of every vertex to an open center that is within the distance of $3r+3$, and the total opening cost of this set is no greater than $B_i$. The union of these sets is the desired solution from the triangle inequality. Recall that $r$ can be taken as two, from Lemma~\ref{l:treealgorithm}.
\end{proof}

\begin{proof}[Proof of Lemma~\ref{l:budgetmain}]
We invoke the rounding procedure given in Section~\ref{sec:gen2tree}, but with the ``fake'' capacity function $\hat L$ defined as $\hat L(v):=\bar C_{\max}-C(v)$, where $\bar C_{\max}:=1+\max_{v\in V}C(v)$. The output vector $y'$ is an integral distance-$(3r+2)$ transfer of $(G,\hat L,y)$ from Lemma~\ref{l:redmain2}. Since $y'$ is a distance-$(3r+2)$ transfer, we have $\sum_{v\in V}y_v = \sum_{v\in V}y'_v = k$, and by taking $U=V$ in Condition~(\ref{def:transfer}\ref{def:transfer2}) of Definition~\ref{def:transfer}, we also have $\sum_{v\in V}\hat L(v)y'_v\geq \sum_{v\in V}\hat L(v)y_v$. Since $\sum_{v\in V}\hat L(v)y'_v=\bar C_{\max}\cdot k-\sum_{v\in V}C(v)y'_v$ and $\sum_{v\in V}\hat L(v)y_v=\bar C_{\max}\cdot k-\sum_{v\in V}C(v)y_v$, this implies that $\sum_{v\in V}C(v)y'_v\leq \sum_{v\in V}C(v)y_v$.

On the other hand, one can see that the decisions made by our rounding procedure purely depend on the relative ordering of capacities, rather than their actual values. Hence, the complete ``execution history'' of the rounding procedure with $\hat L$ could also be interpreted as a valid execution history with the true capacity function $L$ as well: if the procedure is executed with $L$, every comparison of capacities will always be a tie since $L$ is a constant function, and we can break them so that it will be consistent with the ordering of $\hat L$. Therefore, it is possible that our rounding algorithm outputs $y'$ when it is run with $L$, and from Lemma~\ref{l:redmain2}, $y'$ is an integral distance-$(3r+2)$ transfer of $(G,L,y)$.
\end{proof}

\section{6-approximation algorithm for the $\zerol$-case}\label{sec:ap:zerol}
In this section, we present the 6-approximation algorithm for the $\zerol$-case by proving Theorem~\ref{t:zerolmain}. We call a vertex a 0-node if its capacity is zero; an $L$-node otherwise. Let $V_L$ denote the set of $L$-nodes. $\NL(v)$ denotes $\NN(v)\cap V_L$. Let $G=(V,E)$ denote the connected component $G^*_{\leq\tau}$ after the two preprocessing steps described in Section~\ref{sec:better-preprocessing}.

Recall that the 9-approximation algorithm rounds the opening variables of the LP solution ``locally'': it considers the tree of clusters in the bottom-up fashion, and for each subtree $T_u$, it opens $\lfloor y(T_u)\rfloor$ centers while deferring the decision of whether to open one additional center to the later subinstances. Our 6-relaxed decision procedure also operates as a bottom-up local rounding procedure, but in this case, our preprocessing ensures that a path from (the midpoint of) a child cluster to (the midpoint of) the parent does not contain consecutive 0-nodes; this implies that $L$-nodes are very well ``dispersed'' throughout the graph, permitting local rounding to be performed at a finer granularity within closer proximity. In fact, even without such change in the granularity of rounding, a careful choice of $m_v$ alone with the original rounding algorithm is sufficient to give a 8-relaxed decision procedure.

Further improvements are facilitated by a better clustering. The clustering algorithm of Khuller and Sussman~\cite{KhullerS00} that is used by our 9-approximation algorithm finds cluster midpoints that are connected by length-three paths. This is in order to guarantee that $y(C_v)\geq 1$ for each cluster $C_v$, by ensuring $\NN(v)\subset C_v$. However, in a $\zerol$-instance, $\NL(v)\subset C_v$ is sufficient to yield $y(\NN(v)\cap C_v)\geq 1$, and hence we can choose two vertices that are at distance 2 as cluster midpoints as long as all their common neighbors are 0-nodes. This observation leads to an improved clustering where some parent and child can be closer.

\paragraph{Clustering algorithm.}
Our clustering algorithm identifies clusters one by one, and each time a new cluster midpoint $v$ is identified, $\NL(v)$ is allotted to the new cluster $C_v$. The next cluster midpoint is always chosen at distance 2 from the set of already allotted vertices to ensure that $\NL(v)$ of each cluster are disjoint. In what follows, $V_{\textnormal{\textsf{allotted}}}$ denotes the set of vertices that has been already allotted to a cluster by the algorithm; for $u\in V_{\textnormal{\textsf{allotted}}}$, $\alpha(u)$ denotes the midpoint of the cluster that $u$ is allotted to: $u\in C_{\alpha(u)}$; finally, $\textsf{dist}(v)$ denotes the shortest distance from $V_{\textnormal{\textsf{allotted}}}$ to $v$: $\textsf{dist}(v):=\min_{u\in V_{\textnormal{\textsf{allotted}}}} d_G(u,v)$.

Algorithm~\ref{a:cluster} shows our clustering algorithm. In addition to identifying the clusters, our algorithm chooses $p(v)\in C_v$ for each cluster $C_v$, on which the opening of one will be aggregated. Also, for each non-root cluster $C_v$, the algorithm finds a vertex in the parent cluster through which $v$ is connected to the parent cluster and call it $\pi_1(v)$. At the end of the algorithm, we assign every unallotted $L$-node to a nearby cluster; the algorithm annotates each of these vertices with $\pi_2(v)$, where $\pi_2(v)$ denotes the vertex through which $v$ is connected to $\alpha(v)$. 

\begin{algorithm}
\caption{Clustering algorithm.}
\label{a:cluster}
\begin{algorithmic}[1]
	\State $V_{\textnormal{\textsf{allotted}}}\gets\emptyset$\label{st:a:c:1}
	\State Let $v$ be an arbitrary $L$-node\label{st:a:c:2}
	\State Create a new cluster centered at $v$: $C_v\gets \NL(v)$; $V_{\textnormal{\textsf{allotted}}}\gets V_{\textnormal{\textsf{allotted}}}\cup C_v$\label{st:a:c:3}
	\State $p(v)\gets v$\label{st:a:c:4}
	\While{$\exists w\in V_L \ \textsf{dist}(w)\geq 2$}\label{st:a:c:5}
		\State Let $v\in V$ be an arbitrary vertex with $\textsf{dist}(v)=2$\label{st:a:c:6}
		\State $u^*\in\argmin_{u\in V_{\textnormal{\textsf{allotted}}}} d_G(u,v)$\label{st:a:c:7}
		\State Create a new cluster centered at $v$, as a child of $C_{\alpha(u^*)}$:
		\Statex \hspace{1.2em} $C_v\gets \{v\}\cup \NL(v)$; $V_{\textnormal{\textsf{allotted}}}\gets V_{\textnormal{\textsf{allotted}}}\cup C_v$\label{st:a:c:8}
		\State $\pi_1(v)\gets u^*$\label{st:a:c:9}
		\If {$v$ is an $L$-node} \hspace{2pt}$p(v)\gets v$ \textbf{else} $p(v)$ is arbitrarily chosen from $\NN(v)\cap \NN(u^*)$ \label{st:a:c:10}\EndIf
	\EndWhile
	\State $V^*_{\textnormal{\textsf{allotted}}}\gets V_{\textnormal{\textsf{allotted}}}$\label{st:a:c:11}
	\ForAll{$v\in V_L\setminus V^*_{\textnormal{\textsf{allotted}}}$}\label{st:a:c:12}
		\State Let $u$ be an arbitrary vertex in $V^*_{\textnormal{\textsf{allotted}}} \cap \NN(v)$\label{st:a:c:13}
		\State $C_{\alpha(u)}\gets C_{\alpha(u)} \cup\{v\}$\label{st:a:c:14}
		\State $\pi_2(v)\gets u$\label{st:a:c:15}
	\EndFor
\end{algorithmic}
\end{algorithm}

\begin{lemma}\label{l:cl}
Algorithm~\ref{a:cluster} is well-defined, and its output satisfies the following:\begin{enumerate}[(i)]
\item $\NL(v)\subset C_v$ for every $C_v$, and $C_v$'s are disjoint;\label{p:l:cl:1}
\item every $L$-node is allotted to some cluster, and a 0-node is allotted only when it becomes a cluster midpoint;\label{p:l:cl:2}
\item $p(v)\in \NL(v)$ for every $C_v$;\label{p:l:cl:3}
\item $\pi_1(v)$, when defined, is in $\NL(x)$ for some $C_x$; $\pi_2(v)$, when defined, is in $\NL(y)$ for some $C_y$;\label{p:l:cl:range}
\item $C_v=\begin{cases}
\NL(v)\cup\{u\mid\pi_2(u)\in \NL(v)\},&\textrm{if }v\textrm{ is an }L\textrm{-node};\\
\{v\}\cup \NL(v)\cup\{u\mid\pi_2(u)\in \NL(v)\},&\textrm{if }v\textrm{ is a 0-node}.
\end{cases}$\label{p:l:cl:part}
\end{enumerate}
\end{lemma}
\begin{proof}
Since $\LP_k(G)$ is feasible, $V_L\neq\emptyset$ and $v$ can be chosen at Step~\ref{st:a:c:2}. At Step~\ref{st:a:c:6}, as there exists $w\in V_L$ with $\textsf{dist}(w)\geq 2$, there exists a vertex $v$ with $\textsf{dist}(v)=2$, for example the one that appears on a path of length $\textsf{dist}(w)$ from $V_{\textnormal{\textsf{allotted}}}$ to $w$. Note that $v\in V$ may be a 0-node or an $L$-node. At Step~\ref{st:a:c:10}, $d_G(u^*,v)=2$ from the choice of $u^*$ and hence $\NN(v)\cap \NN(u^*)$ is nonempty. When the \mbox{\textbf{while}} loop terminates, $\textsf{dist}(w)\leq 1$ for every $w\in V_L$; thus, $v$ at Step~\ref{st:a:c:13} satisfies $\textsf{dist}(v)=1$ and therefore $u$ can be chosen. The algorithm is well-defined.

Each time a new cluster $C_v$ is created, $\NL(v)$ is added to $C_v$: $\NL(v)\subset C_v$. The only two cases in which we create a new cluster $C_v$ is when it is the first cluster created, and when $\textsf{dist}(v)=2$. In the latter case, since $\textsf{dist}(v)=2$, $\NL(v)\cap V_{\textnormal{\textsf{allotted}}} =\emptyset$ and therefore $\{v\}\cup \NL(v)$ is disjoint from $V_{\textnormal{\textsf{allotted}}}$, the set of already allotted vertices. Thus, at the beginning of Step~\ref{st:a:c:11}, $C_v$'s are disjoint. No new clusters are created in the rest of the algorithm and only the $L$-nodes that has not been allotted are added to exactly one of the existing clusters. Hence, Property~\eqref{p:l:cl:1} holds.

Property~\eqref{p:l:cl:2} is easily verified, since \mbox{Steps~\ref{st:a:c:12}-\ref{st:a:c:15}} ensure that every $L$-node is allotted, and the only case a 0-node is allotted is at Step~\ref{st:a:c:8}, where the cluster midpoint $v$ is allotted.

At Step~\ref{st:a:c:10}, if $v$ is a 0-node, $\NN(v)\subset V_L$ since every edge is incident to at least one $L$-node; Property~\eqref{p:l:cl:3} follows from this observation.

Until Step~\ref{st:a:c:11} of the algorithm, vertices are allotted only when it is a cluster midpoint or in $\NL(v)$ for some cluster midpoint $v$. Thus, $u^*\in V_{\textnormal{\textsf{allotted}}}$ chosen at Step~\ref{st:a:c:7} is either a cluster midpoint or in $\NL(v)$ for some $C_v$. Suppose $u^*$ is a cluster midpoint. If $u^*$ is an $L$-node, then $u^*\in \NL(u^*)$; suppose $u^*$ is a 0-node. As $d_G(u^*,v)=2$ from the choice of $u^*$, there exists a vertex $z$ that is in both $\NN(u^*)$ and $\NN(v)$. $z$ is an $L$-node since $u^*$ is a 0-node. Thus $z$ is in $\NL(u^*)$ and has to be in $V_{\textnormal{\textsf{allotted}}}$, contradicting $\textsf{dist}(v)=2$. Hence, in any case, $\pi_1(v)\in \NL(x)$ for some $C_x$. At Step~\ref{st:a:c:13}, $u\in V^*_{\textnormal{\textsf{allotted}}}$ and hence either $u$ is a cluster midpoint or $u\in \NL(y)$ for some $C_y$. If $u$ is a cluster midpoint, $v\in \NL(u)$, contradicting $v\notin V^*_{\textnormal{\textsf{allotted}}}$. Property~\eqref{p:l:cl:range} is verified.

At the beginning of Step~\ref{st:a:c:11}, for every $C_v$, $C_v=\{v\}\cup \NL(v)$ from construction and it can be only augmented in the rest of the algorithm. When $v$ is added to a cluster at Step~\ref{st:a:c:14}, it is added to $C_{\alpha(\pi_2(v))}$ and hence\[
C_v=\begin{cases}
\NL(v)\cup\{u\mid\pi_2(u)\in \{v\}\cup \NL(v)\},&\textrm{if }v\textrm{ is a }L\textrm{-node};\\
\{v\}\cup \NL(v)\cup\{u\mid\pi_2(u)\in \{v\}\cup \NL(v)\},&\textrm{if }v\textrm{ is a 0-node}.
\end{cases}
\]Now Property~\eqref{p:l:cl:part} follows from Property~\eqref{p:l:cl:range}.
\end{proof}

\begin{observation}\label{o:pip}
For every non-root cluster $C_v$, the distance between $\pi_1(v)$ and $p(v)$ is\[
\begin{cases}
1,&\textrm{if }v\textrm{ is a 0-node};\\
2,&otherwise.
\end{cases}\]
\end{observation}
\begin{proof}
Note that $\pi_1(v)$ and $v$ are at distance 2 as can be seen from Step~\ref{st:a:c:7} of Algorithm~\ref{a:cluster}; thus, if $v$ is an $L$-node, $p(v)=v$ and the distance between $\pi_1(v)$ and $p(v)$ is two. If $v$ is a 0-node, $p(v)$ is chosen from $\NN(\pi_1(v))$ at Step~\ref{st:a:c:10} of Algorithm~\ref{a:cluster}.
\end{proof}

\paragraph{Rounding opening variables.}
Our algorithm will gradually round the opening variables $y$, starting from the original LP solution, until they become integral. This process will be described in terms of \emph{opening movements}, where each movement specifies how much opening is moved from which $L$-node to which $L$-node. Since the $L$-nodes have the same capacities, if we show that a set of opening movements makes the opening variables integral while no opening is moved by the net distance of more than $r$, this implies that the resulting set of opening variables is an integral distance-$r$ transfer.

Our rounding procedure begins with changing $y_{p(v)}$ of every cluster $C_v$ to one: for each cluster $C_v$, we increase $y_{p(v)}$ until it reaches one, while simultaneously decreasing the opening variable of a vertex in $\NL(v)$ by the same amount. This \emph{initial aggregation} can be interpreted as opening movements, and keeps the budget constraint $\sum_{v\in V}y_v=k$ satisfied.

\begin{observation}\label{o:pre}
For each cluster $C_v$, the initial aggregation can be implemented by a set of opening movements within the distance of\[
\begin{cases}
1,&\textrm{if }p(v)=v;\\
2,&otherwise.
\end{cases}\]
\end{observation}
\begin{proof}
$\NL(v)$'s are disjoint from Lemma~\ref{l:cl}, and $y(\NL(v))\geq 1$ from the LP constraints; hence, $y_{p(v)}$ can be made 1 via movements from $\NL(v)$. Note that $p(v)\in \NL(v)$ in any case.
\end{proof}

After the initial aggregation, the procedure considers each cluster $C_v$ in the bottom-up order and make the opening variables of every vertex in $C_v\setminus\{p(v)\}$ integral, using movements of distance 5 or smaller; $p(v)$ is propagated to the parent cluster, to be taken into account when that cluster is rounded. Precisely, the rounding procedure for $C_v$ rounds the opening variables of $I_v:=V_L \cap (C_v\setminus\{p(v)\}\cup\{p(u)\mid\pi_1(u)\in C_v\})$, i.e., the set of $L$-nodes that is either propagated from a child cluster or originally in $C_v$, except the vertex to be propagated \emph{from} $C_v$.

Algorithm~\ref{a:rounding} shows the procedure. First it recursively processes the children clusters, and then constructs a family of vertex sets $\{X_u\}_{u\in \NL(v)}$ indexed by $\NL(v)$. For $u\neq p(v)$, $X_u$ consists of $u$ itself, vertices propagated from the children clusters that are connected through $u$, and the vertices in $C_v$ that are connected to $v$ through $u$: $X_u:=\{u\}\cup\{p(w)\mid\pi_1(w)=u\}\cup\{w\mid\pi_2(w)=u\}$. $X_{p(v)}$ is similarly defined, except that it does not contain $p(v)$. Now for every $u\in \NL(v)$, we locally round $X_u$: we choose a set $W_u$ of the vertices to be opened, and move the openings of the other vertices to the vertices in $W_u$. Note that \textsc{LocalRound}($V_{\textnormal{\textsf{toOpen}}}$, $V_{\textnormal{\textsf{moveFrom1}}}$, $V_{\textnormal{\textsf{moveFrom2}}}$) is a procedure that increases the opening variables of the vertices in $V_{\textnormal{\textsf{toOpen}}}$ to one, while decreasing the opening variables of $V_{\textnormal{\textsf{moveFrom1}}}$ (and $V_{\textnormal{\textsf{moveFrom2}}}$ if $V_{\textnormal{\textsf{moveFrom1}}}$ is used up) to match the increase. $W_u$ is chosen as a subset of $X_u$, but we avoid choosing $u\in \NL(v)$ whenever possible. After these local roundings, each $X_u$ may still have some non-integral opening variables remaining; we choose a set $F\subset \NL(v)\setminus\{p(v)\}$ to accomodate these openings. Finally, if there still remains some fraction, we choose one last center $w^*$, and open it using the opening movements from $\bar I_{\textsf{Step\ref{st:a:r:12}}}$ and $\{p(v)\}$. Note that $y_{p(v)}$, therefore, may become less than one at the termination of \textsc{Round}($v$).

\begin{algorithm}
\caption{Rounding algorithm.}
\label{a:rounding}
\begin{algorithmic}[1]
	\Procedure{Round}{$v$}
		\ForAll {children clusters $C_w$} \Call {Round}{$w$} \EndFor\label{st:a:r:2}
		\State $X_{p(v)}\gets \{p(w)\mid\pi_1(w)=p(v)\}\cup\{w\mid\pi_2(w)=p(v)\}$\label{st:a:r:3}
		\State $X_u\gets \{u\}\cup\{p(w)\mid\pi_1(w)=u\}\cup\{w\mid\pi_2(w)=u\}$ \textbf{for all} $u\in \NL(v)\setminus\{p(v)\}$\label{st:a:r:4}
		\ForAll {$u\in \NL(v)$}\label{st:a:r:5}
		  \State Choose $\lfloor y(X_u)\rfloor$ vertices from $X_u$; call it $W_u$ (avoid choosing $u$ unless $|X_u|=y(X_u)$)\label{st:a:r:6}
			\State \Call{LocalRound}{$W_u$, $X_u$, $\emptyset$}\label{st:a:r:7}
		\EndFor
		\State Let $\bar I_{\textsf{Step\ref{st:a:r:8}}}$ be the set of vertices in $\cup_{u\in \NL(v)}X_u$ that have non-integral opening variables\label{st:a:r:8}
		\State $F:=\{u\in \NL(v)\setminus \{p(v)\}\mid y_u<1 \}$\label{st:a:r:9}
		\State Choose $\lfloor \sum_{u\in \bar I_{\textsf{Step\ref{st:a:r:8}}}} y_u \rfloor$ vertices from $F$; call it $W_F$\label{st:a:r:10}
		\State \Call{LocalRound}{$W_F$, $\bar I_{\textsf{Step\ref{st:a:r:8}}}\setminus X_{p(v)}$, $\bar I_{\textsf{Step\ref{st:a:r:8}}}\cap X_{p(v)}$}\label{st:a:r:11}
		\State Let $\bar I_{\textsf{Step\ref{st:a:r:12}}}$ be the set of vertices in $\cup_{u\in \NL(v)}X_u$ that have non-integral opening variables\label{st:a:r:12}
		\If {$\bar I_{\textsf{Step\ref{st:a:r:12}}}\neq\emptyset$}\label{st:a:r:13}
			\State Choose $w^*$ from $F\setminus W_F$ if $F\setminus W_F\neq\emptyset$; otherwise choose from $\bar I_{\textsf{Step\ref{st:a:r:12}}}$\label{st:a:r:14}
			\State \Call{LocalRound}{$\{w^*\}$, $\bar I_{\textsf{Step\ref{st:a:r:12}}}$, $\{p(v)\}$}\label{st:a:r:15}
		\EndIf
	\EndProcedure\vspace{1ex}
	\Procedure{LocalRound}{$V_{\textnormal{\textsf{toOpen}}}$, $V_{\textnormal{\textsf{moveFrom1}}}$, $V_{\textnormal{\textsf{moveFrom2}}}$}
		\While {$\exists u\in V_{\textnormal{\textsf{toOpen}}} \  y_u<1$}
			\State Choose a vertex $w$ with nonzero opening from $V_{\textnormal{\textsf{moveFrom1}}}\setminus V_{\textnormal{\textsf{toOpen}}}$;\label{st:a:lr:3}
			\Statex \hspace{2.7em} if there exists none, choose from $V_{\textnormal{\textsf{moveFrom2}}}\setminus V_{\textnormal{\textsf{toOpen}}}$
			\State $\Delta\gets \min(1-y_u,y_w)$; increase $y_u$ by $\Delta$ and decrease $y_w$ by $\Delta$
		\EndWhile
	\EndProcedure
\end{algorithmic}
\end{algorithm}

\begin{lemma}\label{l:6key}
Suppose that $y_{p(v)}=1$ before Step~\ref{st:a:r:3} of \textnormal{\textsc{Round}($v$)}. Then Steps~\ref{st:a:r:3}-\ref{st:a:r:15} of \textnormal{\textsc{Round}($v$)} make the opening variables of $I_v$ integral, and this can be implemented by a set of opening movements within $I_v\cup\{p(v)\}$, with no incoming movements to $p(v)$. The maximum distance of these movements is five taking the initial aggregation into account.
\end{lemma}
\begin{proof}
Note that, from Properties~\eqref{p:l:cl:3}, \eqref{p:l:cl:range}, and \eqref{p:l:cl:part} of Lemma~\ref{l:cl}, $\{X_u\}_{u\in \NL(v)}$ forms a partition of $I_v$. Thus it suffices to verify that the opening variables of each $X_u$ becomes integral. Also note that $X_u\subset V_L$.

At Step~\ref{st:a:r:6} of the algorithm, we have $y(X_u)\leq|X_u|$ from $y\leq\textbf{1}$; hence $W_u$ can be successfully chosen. After Step~\ref{st:a:r:7}, $X_u$ may still have some non-integral opening variables, but their total opening is given by $r_u:=y(X_u)-\lfloor y(X_u)\rfloor <1$. Moreover, when $r_u>0$, we have $u\notin W_u$ and therefore $y_u<1$. Thus, $\sum_{u\in\bar I_{\textsf{Step\ref{st:a:r:8}}}}y_u=r_{p(v)}+\sum_{u\in \NL(v)\setminus\{p(v)\}} r_u<1+|F|$ at Step~\ref{st:a:r:10} and therefore $W_F$ can be chosen as well. Note that $(\bar I_{\textsf{Step\ref{st:a:r:8}}}\setminus X_{p(v)})\cup(\bar I_{\textsf{Step\ref{st:a:r:8}}}\cap X_{p(v)})=\bar I_{\textsf{Step\ref{st:a:r:8}}}$ and hence Step~\ref{st:a:lr:3} of \mbox{\textsc{LocalRound}} called from Step~\ref{st:a:r:11} will always succeed. After Step~\ref{st:a:r:11}, the total non-integral opening variables in $I_v$ will become strictly smaller than one. If $\bar I_{\textsf{Step\ref{st:a:r:12}}}=\emptyset$, we are done. Otherwise, $u$ can be successfully chosen since $\bar I_{\textsf{Step\ref{st:a:r:12}}}\neq\emptyset$, and Step~\ref{st:a:r:15} will make the opening variables of $I_v$ completely integral, while making $y_{p(v)}$ smaller than one. Note that $y_{p(v)}=1$ before Step~\ref{st:a:r:15}.

Now it remains to verify that this rounding can be realized in terms of opening movements within the distance of five. When $x\in X_u$, one of the following holds: \ite{i} $x=u$, \ite{ii} $x=p(w)$ and $\pi_1(w)=u$, or \ite{iii} $\pi_2(x)=u$. In Case~\ite{ii}, $d_G(x,u)\leq 2$ from Observation~\ref{o:pip}. In Case~\ite{iii}, $d_G(x,u)=1$ as can be seen from Step~\ref{st:a:c:13} of Algorithm~\ref{a:cluster}. Thus, for any $x\in X_u$, $x$ is within the distance of 2 from $u$. Note that the \emph{opening} at $x\in X_u$ has been moved from $\NL(w)$ if $x=p(w)$; otherwise, it originates from $x$ itself. From Observations~\ref{o:pip} and \ref{o:pre}, the opening at $x$, in Case~\ite{ii}, originates from vertices within the distance of three from $u$; in the other cases, it is from $x$ itself and therefore within the distance of one. Thus, any movements resulting from Step~\ref{st:a:r:7} of Algorithm~\ref{a:rounding} moves opening that originally comes from vertices within the distance of three from $u$ to a vertex within the distance of two from $u$; the maximum distance of these movements therefore is five.

Since the opening at $x\in X_u$ originates from vertices within the distance of three from $u$, it is within the distance of four from $v$. On the other hand, every vertex in $F$ is within the distance of one from $v$; therefore, the maximum distance of movements resulting from Step~\ref{st:a:r:11} also is five.

Suppose $w^*$ is chosen from $F\setminus W_F$ at Step~\ref{st:a:r:14}. Then $w^*$ is within the distance of one from $v$; as observed earlier, the opening at $x\in\bar I_{\textsf{Step\ref{st:a:r:12}}}$ originates from vertices within the distance of four from $v$. $p(v)$ is within the distance of one from $v$, and its opening originates from vertices within the distance of two from $p(v)$ (see Observation~\ref{o:pre}); hence, the opening at $p(v)$ originates from vertices within the distance of three from $v$. Thus, in this case, any movements resulting from Step~\ref{st:a:r:15} moves opening that originates from vertices within the distance of four from $v$ to a vertex within the distance of one from $v$; the maximum distance of these movements therefore is five.

Suppose $F\setminus W_F=\emptyset$. In this case, $\sum_{u\in\bar I_{\textsf{Step\ref{st:a:r:8}}}\setminus X_{p(v)}} y_u<|F|=|W_F|$ and $(W_F \cap \bar I_{\textsf{Step\ref{st:a:r:8}}}) \subset (I_{\textsf{Step\ref{st:a:r:8}}}\setminus X_{p(v)})$; hence, $\bar I_{\textsf{Step\ref{st:a:r:8}}}\setminus X_{p(v)}$ is used up during \mbox{\textsc{LocalRound}} called from Step~\ref{st:a:r:11}. Therefore, we have $\bar I_{\textsf{Step\ref{st:a:r:12}}}\subset X_{p(v)}$. As observed earlier, $w^*\in X_{p(v)}$ is within the distance of two from $p(v)$; the opening at $x\in X_{p(v)}$ is from vertices within the distance of three from $p(v)$. The opening at $p(v)$ is from vertices within the distance of two from $p(v)$, as was seen in Observation~\ref{o:pre}. Thus, the maximum distance of movements resulting from Step~\ref{st:a:r:15} is five in this case as well.
\end{proof}

\begin{proof}[Proof of Theorem~\ref{t:zerolmain}]
Let $C_r$ be the root cluster, and we excute \textsc{Round}($r$) on the LP solution.

From Lemma~\ref{l:6key}, \textsc{Round}($r$) outputs a set of opening variables that can be realized by a set of opening movements of distance five or smaller: note that $y_{p(v)}=1$ before Step~\ref{st:a:r:3} of \textnormal{\textsc{Round}($v$)}, since we process the clusters in the bottom-up order. As every vertex in $V_L\setminus\{p(r)\}$ is in $I_v$ for some cluster $C_v$,  their opening variables are made integral. Since $\sum_{v\in V}y_v=k$, $y_{p(r)}$ is also made integral, and the opening movements to $p(r)$ during the initial aggregation were within the distance of one. Thus the output set of open vertices is an integral distance-5 transfer.

Now Lemma~\ref{lem:framework-main} completes the proof.
\end{proof}

\end{document}